\documentclass[runningheads, thm-restate]{llncs}
\usepackage[T1]{fontenc}
\usepackage{algorithm}
\usepackage{algpseudocode}
\usepackage{amssymb}
\usepackage{amsmath}
\usepackage{hyperref}
\usepackage{graphicx}
\usepackage{color}

\begin{document}
\title{Exact recovery of planted cliques in semi-random graphs}
%
%
\author{Yash Khanna}
\authorrunning{Yash}
%
\institute{Indian Institute of Science, Bangalore, India\footnote{This work was done while the author was a student at IISc.}\\
\email{ykhannay@gmail.com}}
\maketitle              
%
\newcommand{\inprod}[1]{\left\langle #1\right\rangle}
\newcommand{\paren}[1]{\left(#1 \right )}
\newcommand{\defeq}{\stackrel{\textup{def}}{=}}
\newcommand{\clique}{{\sf~Clique}}
\newcommand{\kclique}{{\sc k-clique} problem}
\newcommand{\cliqueparams}{\clique$(n, k, p, r, s, t, d, w, \gamma, \lambda)$}
\newcommand{\cliqueprob}{{\sc Planted Clique} problem}
\newcommand{\dksprob}{{\sc Densest $k$-subgraph} problem}
\newcommand{\dks}{{\sc D$k$S}}
\newcommand{\dksregparams}{\dks{\sf ExpReg}$(n, k, d, \delta, d', \lambda)$}
\newcommand{\dkssregparams}{\dks{\sf Reg}$(n, k, d, \delta, \gamma)$}
\newcommand{\set}[1]{\left\{#1\right\}}
\newcommand{\brac}[1]{[#1 ]}
\newcommand{\abs}[1]{\left\lvert#1\right\rvert}
\newcommand{\Abs}[1]{\left\lvert#1\right\rvert}
\newcommand{\Norm}[1]{\left\lVert#1\right\rVert}
\newcommand{\norm}[1]{\left\lVert#1\right\rVert}
\newcommand{\normt}[1]{\norm{#1}_{\scriptstyle 2}}
\newcommand{\sdp}{{\sf SDP }}
\newcommand{\bigO}{\mathcal{O}}
\newcommand{\flatfrac}[2]{#1/#2}
\newcommand{\ffrac}{\flatfrac}
\newcommand{\etal}{et al. }
\newtheorem{SDP}[theorem]{SDP}
\newtheorem{LP}[theorem]{LP}
\newcommand{\probabilityconstant}{\kappa}
\newcommand{\matrixconstant}{\xi}
\newcommand{\OPT}{{\sf OPT}}
\newcommand{\notimplies}{\;\not\!\!\!\implies}
\begin{abstract}
In this paper, we study the \cliqueprob~in a semi-random model. Our model is inspired from the Feige-Kilian model \cite{FEIGE2001639} which has been studied in many other works \cite{10756085,10.1145/3564246.3585184,article,10.1007/978-3-030-83508-8_38,10.5555/3381089.3381134,DBLP:journals/eccc/Steinhardt17} for a variety of graph problems. Our algorithm and analysis is on similar lines to the one studied for the \dksprob~in the work of Khanna and Louis \cite{khanna_et_al:LIPIcs.FSTTCS.2020.27}. 
	
As a by-product of our main result, we give an alternate SDP-based rounding algorithm (with similar guarantees) for solving the \cliqueprob~in a random graph.

\keywords{Planted cliques   \and Semi-random models \and Beyond worst-case analysis}
\end{abstract}
\section{Introduction}
\label{sec:intro}
Given an undirected graph, the decision problem of checking whether it contains a $k$-clique, i.e., a subgraph of size $k$ which contains all the possible edges is a famous NP-hard problem and appears in the list of 21 NP-complete problems in the early work of Karp \cite{DBLP:conf/coco/Karp72}. The best known approximation algorithm by the work of Boppana and Halld\'{o}rsson \cite{10.1007/3-540-52846-6_74} has an approximation factor of $\bigO\paren{\ffrac{n}{\paren{\log n}^2}}$. The results by H{\aa}stad and Zuckerman \cite{548522,10.1145/1132516.1132612} shows that no polynomial time algorithm can approximate this to a factor better than $n^{1-\epsilon}$ for every $\epsilon > 0$, unless $P = NP$. This was improved by Khot \etal \cite{10.1007/11786986_21}, who showed that there is no algorithm which approximates the maximum clique problem (in the general case) to a factor better than $n/{2^{\paren{\log n}}}^{3/4+\epsilon}$ for any constant $\epsilon > 0$ assuming $NP \subsetneq BPTIME\paren{2^{\paren{\log n}^{\bigO(1)}}}$.

These results led to studying this problem in the average-case, i.e., we plant a clique of size $k$ in a Erd\H{o}s-R\'enyi random graph $(G(n, p))$, and study the ranges of parameters of $k$ and $p$ for which this problem can be solved. We give a brief survey in Section \ref{sec:related_work}.

Another direction is to consider the problem in a restricted family of graphs or ``easier'' instances. This allows us to design new and interesting algorithms with much better guarantees (as compared to the worst-case models) and might possibly help us get away from the adversarial examples which cause the problem to be hard in the first place. This way of studying hard problems falls under the area of ``Beyond worst-case analysis''. We take this approach and in this work, we study the \cliqueprob~in a semi-random model. This is a model generated in multiple stages via a combination of adversarial and random steps. Such generative models have been studied in the early works of \cite{10.1006/jagm.1995.1034,10.1016/j.jalgor.2004.07.003,FEIGE2001639,article} in the context of algorithms. We refer the reader to \cite{khanna_et_al:LIPIcs.FSTTCS.2020.27} and the references therein for a survey of variety of graph problems which have been subjected to such a study.

We start by establishing some notation used throughout the paper. 

\subsection{Notation (from \cite{khanna_et_al:LIPIcs.FSTTCS.2020.27})}
\label{sec:notation}
Let $\bar{A}$ denote the adjacency matrix of our input graph $\mathcal{G} = (\mathcal{V}, \mathcal{E})$ whose construction is defined in Section \ref{sec:model}. We use $n \defeq \abs{\mathcal{V}}$, and use $\mathcal{V}$ and $[n] \defeq \set{1,2,\hdots,n}$ interchangeably. 
We assume, w.l.o.g., that $\mathcal{G}$ is a complete graph: if $\set{i,j} \notin \mathcal{E}$, we add $\set{i,j}$
to $\mathcal{E}$ and set $\bar{A}_{ij} = \bar{A}_{ji} = 0$.

For $\mathcal{V'} \subseteq \mathcal{V}$, we use $\mathcal{G}\brac{\mathcal{V'}}$ to denote the subgraph induced on $\mathcal{V'}$.
For a vector $v$, we use $\norm{v}$ to denote $\normt{v}$.
For a matrix $M$, we use $\norm{M}$ to denote the spectral norm, $\norm{M} \defeq \max\limits_{x \neq 0} \dfrac{\norm{Mx}}{\norm{x}}$.

We define probability distributions $\mu$ over finite sets $\Omega$. For a random variable (r.v.) $X : \Omega \to {\rm I\!R}$, its expectation is denoted by ${\rm I\!E}_{x \sim \mu}\brac{X}$. In particular, we define the distribution which we use next. For a vertex set $\mathcal{V'}\subseteq \mathcal{V}$, we define a probability (uniform) distribution $(f_{\mathcal{V'}})$ on the vertex set $\mathcal{V'}$ as follows. For a vertex $i \in \mathcal{V'}$, $f_{\mathcal{V'}}(i) = \dfrac{1}{\abs{\mathcal{V'}}}$. We use $i \sim \mathcal{V'}$ to denote $i \sim f_{\mathcal{V'}}$ for clarity.

\begin{definition}[Restatement of Definition 1.10 from \cite{khanna_et_al:LIPIcs.FSTTCS.2020.27}]
    \label{def:expander}
    A graph $\mathcal{H} = (\mathcal{V}_{\mathcal{H}}, \mathcal{E}_{\mathcal{H}})$ is said to be a $(s, d,\lambda)$-expander if $\abs{\mathcal{V_{\mathcal{H}}}} = s$, $\mathcal{H}$ is $d$-regular, and $\Abs{\lambda_i} \leq \lambda$, $\forall i \in [s]\setminus \set{1}$, where $\lambda_1 \geq \lambda_2 \hdots \geq \lambda_{s}$ are the eigenvalues of the adjacency matrix of $\mathcal{H}$.
\end{definition}
\subsection{Model}
\label{sec:model}
In this section, we describe our semi-random model. We first describe it informally. We start with an empty graph on $n$ vertices and partition it arbitrarily into sets $\mathcal{S}$ and $\mathcal{V} \setminus \mathcal{S}$ of sizes $k$ and $n-k$ respectively. We plant a clique onto the subgraph induced on $\mathcal{S}$ (denoted by $\mathcal{G}\brac{\mathcal{S}}$). The bipartite subgraph $\mathcal{G}\brac{\mathcal{S} \times \mathcal{V} \setminus \mathcal{S}}$ is a random subgraph with parameter $p$, i.e., each edge is added independently with probability $p$. And finally the subgraph $\mathcal{G}\brac{\mathcal{V} \setminus \mathcal{S}}$ is composed of multiple small subgraphs each of which is far from containing a clique of size $k$, and these subgraphs are connected by random independent edges again with parameter $p$. There are three kinds of subgraphs in $\mathcal{G}\brac{\mathcal{V} \setminus \mathcal{S}}$, first we have $r$ disjoint $(s, d, \lambda)-$expander graphs (see Definition \ref{def:expander}) and, second we have $t$ disjoint subgraphs each having an average degree of at most $\gamma k$, and third we have a random graph of size $w$ and parameter $p$. A formal definition is presented below.
\begin{definition}
    \label{def:model}
    An instance of our input graph \\$\mathcal{G} = (\mathcal{V}, \mathcal{E})\sim$\cliqueparams~is generated as follows,
    \begin{enumerate}
        \item \label{step:one} 
            We divide the vertex set $\mathcal{V}~(\abs{\mathcal{V}} = n)$  into two sets, $\mathcal{S}$ and $\mathcal{V} \setminus \mathcal{S}$ with $\abs{\mathcal{S}} = k$. We further divide $\mathcal{V} \setminus \mathcal{S}$ into sets $\Lambda$, $\Pi$, $\Gamma$ such that
    	\begin{itemize}
        	\item The set $\Lambda$ is arbitrarily divided into disjoint subsets $\Lambda_1, \Lambda_2, \hdots, \Lambda_r$ such that for all $\ell \in [r],~\abs{\Lambda_\ell} = s$,
        	\item the set $\Pi$ is arbitrarily divided into disjoint subsets $\Pi_1, \Pi_2, \hdots, \Pi_t$ such that for all $\ell \in [t],~\abs{\Pi_\ell} > 0$, and
            \item the set $\Gamma$ is such that it has size $\abs{w} > 0$.
    	\end{itemize}
    \item \label{step:two} (\emph{Adding random edges}) We add edges between the following sets of pairs
	\begin{itemize}
            \item $\mathcal{S} \times \mathcal{V} \setminus \mathcal{S}$,
            \item $\Lambda_i \times \Lambda_j$ for $i, j \in [r], i \neq j$,
            \item $\Pi_i \times \Pi_j$ for $i, j \in [t], i \neq j$,
            \item $\Lambda_i \times \Pi_j$ for $i \in [r], j \in [t]$,
            \item $\Lambda_i \times \Gamma$ for $i \in [r]$,
            \item $\Gamma \times \Pi_j$ for $j \in [t]$
	\end{itemize}
	independently with probability $p$. The edges between pairs of vertices in $\Gamma$ are also added with probability $p$.
    \item \label{step:three} (\emph{Adding a clique on $\mathcal{S}$})
        We add edges between pairs of vertices in $\mathcal{S}$ such that the graph induced on $\mathcal{S}$ is a \emph{clique}. For the sake of brevity, we also add a self loop on each vertex of $\mathcal{V}$, this will make the arithmetic cleaner (like the average degree of $\mathcal{G}[\mathcal{S}]$ is now $k$ instead of $k-1$) and has no severe consequences.
    \item \label{step:five} (\emph{Adding edges in $\Lambda_i$'s})
	For each $i \in [r]$, we add edges between arbitrary pairs of vertices in $\Lambda_i$, such that the graph induced on $\Lambda_i$ is a $(s, d, \lambda)$-expander graph.
    \item \label{step:six} (\emph{Adding edges in $\Pi_i$'s})
	For each $i \in [t]$, we add edges between arbitrary pairs of vertices in $\Pi_i$, such that the graph induced on $\Pi_i$ has the following property,
	$ \max\limits_{\mathcal{V'} \subseteq \Pi_i}\left\{\dfrac{\sum\limits_{i, j \in \mathcal{V'}}\bar{A}_{ij}}{2\abs{\mathcal{V'}}}\right\} \leq \gamma k.$ Or, in other words, for each $i \in [t]$ and $\mathcal{V'} \subseteq \Pi_i$, the maximum average degree of the subgraph $\mathcal{G}\brac{\mathcal{V'}}$ is at most $\gamma k$ for some $\gamma \in (0, 1)$.
    \item \label{step:seven} (\emph{Monotone adversary step}) 
        Arbitrarily delete any of the edges added in Steps \ref{step:two}, \ref{step:five}, or \ref{step:six}.
    \item \label{step:eight}
	Output the resulting graph.
\end{enumerate}
\end{definition}
Note that in our model (Definition \ref{def:model}), the three kinds of subgraphs $\mathcal{G}[\Lambda_i], \mathcal{G}[\Pi_j]$, and $\mathcal{G}[\Gamma]$ which constitute ${\mathcal{G}}\brac{{\mathcal{V}} \setminus {\mathcal{S}}}$ have sparse induced subgraphs by definition (at least in the range of parameters where we study them). It is interesting to see that the first two of them are pairwise exclusive in the sense that a $\mathcal{G}[\Lambda_i]$ graph need not qualify to be $\mathcal{G}[\Pi_j]$ and vice versa. It is an easy exercise to show this.

In this paper, the problem which we study is as follows: Given a graph generated from the above described model, the goal is to recover the planted clique $(\mathcal{G}\brac{\mathcal{S}})$ with high probability. We show that for a ``large'' range of the input parameters, we can indeed solve this problem.

The key ingredient of our algorithm is the following semidefinite program (SDP \ref{sdp:dks}) which is a standard relaxation of the \kclique~(We define the \kclique~as the problem of finding a clique of size $k$, given a large graph as input), however we state it in full for completeness.
\begin{SDP}
\label{sdp:dks}
\begin{align}
    \max_{\set{\set{\bar{X_i}}_{i=1}^{n}, \bar{I}}}\qquad
    \label{eq:sdp1}
    \sum\limits_{i,j=1}^{n}  \bar{A}_{ij}\inprod{\bar{X_i}, \bar{X_j}}& \\
    \text{subject to}\qquad\qquad~
    \label{eq:sdp2}
    \sum\limits_{i=1}^{n} \inprod{\bar{X_i}, \bar{X_i}} &= k \\
    \label{eq:sdp4}
    \sum\limits_{j=1}^{n} \inprod{\bar{X_i}, \bar{X_j}} &\leq k\inprod{\bar{X_i}, \bar{X_i}} & \forall i \in [n]\\
    \label{eq:sdp3}
    \inprod{\bar{X_i}, \bar{X_j}} &= 0 &\forall \paren{i, j} \notin \mathcal{E}\\
    \label{eq:sdp5}
    0 \leq	\inprod{\bar{X_i}, \bar{X_j}} &\leq \inprod{\bar{X_i}, \bar{X_i}} & \forall i, j \in [n],\ (i \neq j)\\
    \label{eq:sdp6}
    \inprod{\bar{X_i}, \bar{X_i}} &\leq 1 & \forall i \in [n]\\
    \label{eq:sdp7}
    \inprod{\bar{X_i}, \bar{I}} &= \inprod{\bar{X_i}, \bar{X_i}} & \forall i \in [n]\\
    \label{eq:sdp8}
    \inprod{\bar{I}, \bar{I}} &= 1\\
    \label{eq:sdp9}
    \bar{X_i} &\in {\rm I\!R}^{n+1} & \forall i \in [n]\\
    \label{eq:sdp10}
    \bar{I} &\in {\rm I\!R}^{n+1}
\end{align}
\end{SDP}
The above SDP can be solved upto arbitrary precision in polynomial time using the Ellipsoid algorithm to fetch the solution set $\set{\set{X_i}_{i=1}^{n}, I}$. It is easy to see that since $\mathcal{G}\brac{\mathcal{S}}$ is a clique, the integral solution corresponding to $\mathcal{S}$ does satisfy the above constraints. We state it now,
\[
    \forall i\in [n], \bar{X_{i}} = \begin{cases}
    \hat{v} & i \in \mathcal{S} \\
    \hat{0} & i \in \mathcal{V} \setminus \mathcal{S}\\
    \end{cases}
    \qquad\text{and}\qquad \bar{I} = \hat{v}
\]
where $\hat{v}$ is any unit vector, $\hat{0}$ is the all zeroes vector, and this feasible solution gives an objective value of $k^2$.

Note that this semidefinite programming relaxation is quite similar to that of \dksprob~from \cite{khanna_et_al:LIPIcs.FSTTCS.2020.27} but we also add the following set of constraints to it,
\begin{align}
    \label{eq:extra_constraint}
    {\inprod{X_i, X_j} = 0 \quad \forall \paren{i, j} \notin \mathcal{E}}.
\end{align}
This is a key difference as compared to the \dksprob~and we will use the above set of constraints crucially in our analysis, much of which is inspired from \cite{khanna_et_al:LIPIcs.FSTTCS.2020.27}. We will describe this in more detail in Section \ref{sec:recovery} and Appendix \ref{sec:analysis}. 

\subsection{Main Result}
\label{sec:result}
We propose an algorithm which is based on rounding the above described SDP \ref{sdp:dks}. The algorithm and the analysis uses tools from the recent literature. Roughly speaking, the ranges of parameters where our algorithm works is when the subgraph $\mathcal{G}\brac{\mathcal{S}}$ is a clique while any other $k$-sized induced subgraph is ``far'' from containing a clique. An advantage of using SDP-based algorithms is that they are robust against a monotone adversary (Step \ref{step:seven} of the model construction). This is an important point because many of the algorithms based on spectral or combinatorial methods are not always robust and may not work effectively with the presence of such adversaries.

\begin{theorem}
    \label{thm:main}
    There exist universal constants $\kappa, \xi \in {\rm I\!R}^{+}$ and a deterministic polynomial time algorithm, which takes an instance of \\\cliqueparams~where \[\nu = \dfrac{36\xi^2(np)(r+t+2)}{k^2\paren{1-6p-2\gamma-\dfrac{d}{s}-\dfrac{\lambda}{k}}^2},\] satisfying $\nu \in (0,1)$, and $p \in [\kappa \ffrac{\log n}{n}, 1)$, and recovers the planted clique $\mathcal{S}$ with high probability (over the randomness of the input).
\end{theorem}
Note that our result does not depend on the size of the subgraphs $\mathcal{G}\brac{\Pi_\ell}$'s but only on their counts, i.e. parameter $t$. Even our model is not parameterized by the sizes of $\Pi_\ell$'s. In other words, all the $\Pi_\ell$'s can be of different sizes but as long as the average degree requirement of subgraphs $\mathcal{G}\brac{\Pi_\ell}$'s (the one stated in Step \ref{step:six}) is met, our result holds.

We see some interesting observations from Theorem \ref{thm:main}. Firstly, there are a few conditions for the algorithm to work,
\begin{enumerate}
	\item $p = \Omega\paren{\dfrac{\log n}{n}}$, or to be verbose, $p$ should be ``large''.
	\item The function $\nu$ (which is dependent on the input parameters) should lie in the range $(0, 1)$, or stated in other words, $\nu$ should be ``small''.
\end{enumerate}
 A setting of input parameters when the value of $\nu$ is ``small'' is as follows:
\[
k = \Omega\paren{\max\paren{\sqrt{np(r+t+2)}, \lambda}}, \gamma = \bigO(1), s=\Omega(d).
\]
\noindent
And also, $6p + 2\gamma + \frac{d}{s} + \frac{\lambda}{k} < 1$.
The above values of different input parameters suggest that the algorithm will work only when any subgraph of size $k$ will be far from a dense set (or a clique) inside $\mathcal{G}\brac{\mathcal{V} \setminus \mathcal{S}}$. Also since the subgraph $\mathcal{G}\brac{\mathcal{S} \times \mathcal{V} \setminus \mathcal{S}}$ is a random graph, it will not have dense sets either, thus naturally we can think of that the SDP \ref{sdp:dks} should put most of the its mass on the vertices of $\mathcal{S}$.
\subsection{Related Work}
\label{sec:related_work}
\paragraph*{Random models for the clique problem.}
For the Erd\H{o}s-R\'enyi~random graph: $G(n, 1/2)$, it is known that the largest clique has a size approximately $2\log_{2}n$ \cite{matula}. There are several poly-time algorithms which find a clique of size $\log_{2}n$, i.e., with an approximation factor roughly $1/2$ \cite{Grimmett_McDiarmid_1975}. It is a long standing open problem to give an algorithm which finds a clique of size $(1+\epsilon)\log_{2}n$ for any fixed $\epsilon > 0$. This conjecture has a few interesting cryptographic consequences as well \cite{10.5555/314613.315047}.

\paragraph*{Planted models for the clique problem.}
In the \cliqueprob, we plant a clique of size $k$ in $G(n, 1/2)$ and study the ranges of $k$ for which this problem can be solved. The work by \cite{KUCERA1995193} shows that if $k = \Omega(\sqrt{n \log n})$, then the planted clique essentially comprises of the vertices of the largest degree. Alon, Krivelevich, and Sudakov \cite{10.5555/314613.315014} give a spectral algorithm to find the clique when $k = \Omega(\sqrt{n})$. Feige and Krauthgamer \cite{article} gave a SDP-based algorithm (different than ours) based on the Lov\'asz~theta function that works for $k = \Omega(\sqrt{n})$ in the presence of a monotone adversary, which can remove the random edges but not the edges of the planted clique. There is also a nearly linear time algorithm which succeeds w.h.p. when $k \geq (1+\epsilon)\sqrt{n/e}$ for any $\epsilon > 0$ \cite{10.1007/s10208-014-9215-y}. When $k = o\paren{\sqrt{n}}$, the work by Barak \etal\cite{doi:10.1137/17M1138236} rules out the possibility for a sum of squares algorithm to work.

For $r = t = 0$, i.e., the case when there are no such $\Lambda_\ell$'s, and $\Pi_\ell$'s. The case when $\mathcal{G}\brac{\mathcal{V} \setminus \mathcal{S}}$ is nothing but a random graph on $n-k$ vertices and probability parameter $p$, the lower bound on $k$ translates to $\Omega(\sqrt{np})$. Thus in this case, our problem reduces to recovering the planted clique in a random graph and we get a similar threshold value of $k$ to the one already studied in literature \cite{article,khanna_et_al:LIPIcs.FSTTCS.2020.27}.

We compare our work with that of recent work on \cliqueprob~by B\l{}asiok \etal\cite{10756085}, and Buhai \etal\cite{10.1145/3564246.3585184} in the table below.
\begin{center}
\begin{tabular}{ |p{3cm}||p{3cm}|p{3cm}|p{3cm}|  }
 \hline
 \multicolumn{4}{|c|}{Comparison with recent work on the \cliqueprob~for $p = 1/2$.} \\
 \hline
  & This work &B\l{}asiok \etal\cite{10756085}&Buhai \etal\cite{10.1145/3564246.3585184}\\
 \hline
 1) Size of k   & \boldmath{$\Omega(\sqrt{n})$}    &$\Omega(\sqrt{n}\log^2 n)$&   $\Omega(n^{1/2+\epsilon})$\\
 \hline
 2) Structure of $\mathcal{V} \setminus \mathcal{S}$&  Union of disjoint sparse graphs  & \textbf{Arbitrary graph}   & \textbf{Arbitrary graph}\\
 \hline
 3) Monotone deletions &\textbf{Allowed} & Not Allowed&  \textbf{Allowed}\\
 \hline
 4) Recovery (w.h.p.)    &\textbf{Exact recovery of $\mathcal{S}$} & List-decoding&  List-decoding\\
 \hline
  5) Running Time    &\textbf{poly($n$)} & \textbf{poly($n$)}&  $n^{\bigO({1/\epsilon})}$\\
 \hline
\end{tabular}    
\end{center}

There are many applications of the \cliqueprob~(and its variants), here is a partial list of works which talk about this problem: \cite{10.1145/1941487.1941511,austrin2011inapproximabilitynpcompletevariantsnash,10.1214/13-AOS1127,bhaskara2024robustnessspectralalgorithmssemirandom,10.1007/978-3-031-93112-3_20,10.1023/A:1008374125234,Koiran2011OnTC,doi:10.1126/science.298.5594.824}.

\paragraph*{Semi-random models for related problems.}
The semi-random model studied in this paper is inspired from a combination of two works. First is the very generic Feige-Kilian model \cite{FEIGE2001639}. In this model, we plant an independent set on $\mathcal{S}$ ($\abs{\mathcal{S}} = k$), the subgraph $\mathcal{G}\brac{\mathcal{S} \times \mathcal{V} \setminus \mathcal{S}}$ is a random graph with parameter $p$, while the subgraph $G\brac{\mathcal{V} \setminus \mathcal{S}}$ can be an arbitrary graph. Then an adversary is allowed to add edges anywhere without disturbing the planted independent set. McKenzie, Mehta, and Trevisan \cite{10.5555/3381089.3381134} show that for $k=\Omega\paren{\ffrac{n^{2/3}}{p^{1/3}}}$, their algorithm finds a ``large'' independent set. And for the range $k=\Omega\paren{\ffrac{n^{2/3}}{p}}$ , their algorithm outputs a list of independent sets (this type of algorithms' output is called the list-decoding variant), one of which is $\mathcal{S}$ with high probability. In the hypergraph case, the work by Khanna \etal~\cite{10.1007/978-3-030-83508-8_38} generalises their results to $r$-uniform hypergraphs for an analogous family of instances. Restrictions of this model has also been studied in the works of \cite{10.1145/3055399.3055491,DBLP:journals/eccc/Steinhardt17}.

It is important to note that while the above model is a pretty generic model and also solves the semi-random model which we study in our paper, however there are some key differences. Firstly, ours is an exact deterministic algorithm based on the SDP relaxation of the \kclique~while they use a ``crude'' SDP (this idea was introduced in \cite{10.5555/3381089.3381134}) which is not a relaxation of the independent set (or the complementary clique problem). But both the SDPs ``clusters'' the vectors corresponding to the planted set. Secondly the algorithmic guarantee of the work by  \cite{10.1007/978-3-030-83508-8_38,10.5555/3381089.3381134} is of a different nature where they output a list of independent sets one of which is the planted set, as described above.

The second relevant model is studied by Khanna and Louis \cite{khanna_et_al:LIPIcs.FSTTCS.2020.27} for the \dksprob. They plant an arbitrary dense subgraph on $\mathcal{G}\brac{\mathcal{S}}$, the subgraph $\mathcal{G}\brac{\mathcal{S} \times \mathcal{V} \setminus \mathcal{S}}$ is a random subgraph, and the subgraph $\mathcal{G}\brac{\mathcal{V} \setminus \mathcal{S}}$ has a property (Step \ref{step:three} of model construction) like the one of $\Lambda_\ell$'s or $\Pi_\ell$'s of this paper. A monotone adversary can delete edges outside $\mathcal{G}\brac{\mathcal{S}}$. Our algorithm, model, and the analysis is inspired from their work. We study the problem in the case when $\mathcal{G}\brac{\mathcal{S}}$ is a clique on $k$ vertices instead of an arbitrary $d$-regular graph. We get a full recovery of the clique in this paper instead of a ``large'' recovery of the planted set, for a ``wide'' range of input parameters.

We now compare our results to the models of Khanna and Louis \cite{khanna_et_al:LIPIcs.FSTTCS.2020.27}.
\begin{itemize}        	
    \item
	Recall the model, \dkssregparams~introduced in their work. In this model, the subgraph $\mathcal{G}\brac{\mathcal{S}}$ is an arbitrary $d$-regular graph of size $k$, $\mathcal{G}\brac{\mathcal{S} \times \mathcal{V} \setminus \mathcal{S}}$ is a random graph with parameter $p$, and the subgraph $\mathcal{G}\brac{\mathcal{V} \setminus \mathcal{S}}$, has the following property,
    $\max\limits_{\mathcal{V'} \subseteq \mathcal{V} \setminus \mathcal{S}}\left\{\dfrac{\sum\limits_{i, j \in \mathcal{V'}}\bar{A}_{{ij}}}{2\abs{\mathcal{V'}}}\right\} \leq \gamma d.$
	Clearly, this is analogous to the case when we only have one such $\Pi_1$ comprising the whole of $\mathcal{G}\brac{\mathcal{V} \setminus \mathcal{S}}$ such that the maximum average degree of any subgraph of $\mathcal{G}\brac{\Pi_1}$ is at most $\gamma k$. Now when $r =  0$ and $t = 1$, our model reduces to the case when \dkssregparams~has a clique on $\mathcal{S}$ (instead of a $d$-regular subgraph). Note that this case can be solved using our algorithm efficiently and we can recover the planted clique, i.e., $\mathcal{S}$ w.h.p. This is a much stronger guarantee as compared to the one in \cite{khanna_et_al:LIPIcs.FSTTCS.2020.27} where they output a vertex set with a large intersection with the planted set (but not completely), with the same threshold on $k$, i.e., $k = \Omega\paren{\sqrt{np}}$.
    \item
	Similarly, in the model, \dksregparams~introduced in the work of \cite{khanna_et_al:LIPIcs.FSTTCS.2020.27}. In this model, the subgraph $\mathcal{G}\brac{\mathcal{S}}$ is an arbitrary $d$-regular graph of size $k$, $\mathcal{G}\brac{\mathcal{S} \times \mathcal{V} \setminus \mathcal{S}}$ is a random graph with parameter $p$, and the subgraph $\mathcal{G}\brac{\mathcal{V} \setminus \mathcal{S}}$, is a $(n-k, d', \lambda)$-expander graph. This is analogous to the case when we have only one such $\Lambda_1$ comprising the whole of $\mathcal{G}\brac{\mathcal{V} \setminus \mathcal{S}}$. Now when $t =  0$ and $r = 1$, our model reduces to the case when \dksregparams~has a clique on $\mathcal{S}$. And similar to the previous point, this case can also be solved using our algorithm efficiently and we can recover the planted clique, i.e., $\mathcal{S}$ w.h.p.
\end{itemize}

\begin{remark}
It is important to note that it has been pointed to us by anonymous reviewers that the partial recovery in the work of Khanna and Louis \cite{khanna_et_al:LIPIcs.FSTTCS.2020.27} as described above can be translated to the full recovery (in the clique case) easily by combining some results from Buhai~\etal\cite{10.1145/3564246.3585184}.
\end{remark}

The idea of using SDP-based algorithms for solving semi-random models of instances has been explored in multiple works for a variety of graph problems, some of which are \cite{10.1145/1806689.1806719,khanna_et_al:LIPIcs.FSTTCS.2020.27,10.1007/978-3-030-83508-8_38,6108205,louis_et_al:LIPIcs.ICALP.2018.101,louis_et_al:LIPIcs.FSTTCS.2019.23,10.1145/2591796.2591841,10.5555/3381089.3381134,10.1145/2746539.2746603}.

\subsection{Proof Idea}
Our algorithm is based on rounding a SDP solution. The basic idea is to show that the vectors corresponding to the planted set $\mathcal{S}$ are ``long''. In the integral solution we have exactly $k$ long vectors which correspond to the set $\mathcal{S}$, in our solution we show that the vertices corresponding to the ``long'' vectors form a subset of $\mathcal{S}$ (the planted clique). This is shown by bounding the contribution of the vectors towards the SDP mass from the rest of the graph (i.e., everything except $\mathcal{G}\brac{\mathcal{S}}$). The decomposable nature of the SDP objective (Equation \ref{eq:sdp1}) into multiple sums corresponding to different subgraphs is leveraged here. This allows us to exploit the geometry of vectors to recover a part of the planted clique. This is possible only because the subgraph $\mathcal{G}\brac{\mathcal{V}\setminus \mathcal{S}}$, which is a combination of expanders, low-degree graphs etc. and thus is a sparse graph by construction (Steps \ref{step:five} or \ref{step:six}) and the random bipartite subgraph $\mathcal{G}\brac{\mathcal{S} \times \mathcal{V} \setminus \mathcal{S}}$ (Step \ref{step:two} of the model construction) will not have any dense sets either. Thus qualitatively the SDP should put most of the mass on the vertices of $\mathcal{S}$. We study the range of input parameters when this happens. 

Once we have recovered a subset of $\mathcal{S}$, the rest of the vertices can be recovered using a greedy algorithm. Let $\mathcal{T}$ denote the set of long vectors obtained by rounding the SDP \ref{sdp:dks} such that $\mathcal{T} \subseteq \mathcal{S}$, the remaining vertices of $\mathcal{S}$ can be obtained by iterating over vertices in $\mathcal{V} \setminus \mathcal{T}$ and checking if it has an edge with all vertices of $\mathcal{T}$ and completing the clique this way. Note that for this to work, we crucially use the orthogonality constraints added for each non-edge pair (Equation \ref{eq:extra_constraint}) and this additional recovery step works only because the planted set is a \kclique~and not an arbitrary dense subgraph. The recovery procedure is explained in Section \ref{sec:recovery} of the paper.

\subsection{Action of Monotone Adversary}
A monotonicity argument can be used to ignore the action of the adversary as stated below.
\begin{lemma}
    \label{lem:adversary}
    In the upcoming discussion and analysis, w.l.o.g., we can ignore the adversarial action (Step \ref{step:seven} of the model construction) to have taken place.
\end{lemma}
\begin{proof}
    Let us assume the monotone adversary removes edges arbitrarily from the subgraphs $\mathcal{G}[\mathcal{V} \setminus \mathcal{S}],~\mathcal{G}[\mathcal{S}, \mathcal{V} \setminus \mathcal{S}]$ and the new resulting adjacency matrix is $\bar{A}$. Then for any feasible solution $\set{\set{Y_i}_{i=1}^{n}, I_Y}$ of the SDP \ref{sdp:dks}, we have $\sum\limits_{i \in P, j \in Q} \bar{A}_{ij}\inprod{Y_i, Y_j} \leq \sum\limits_{i \in P, j \in Q} A_{ij} \inprod{Y_i, Y_j}$ for $\forall P, Q \subseteq \mathcal{V}$. This holds because of the non-negativity SDP constraint \ref{eq:sdp5}. Thus the upper bounds on SDP contribution by vectors in $\mathcal{G}\brac{\mathcal{S}, \mathcal{V} \setminus \mathcal{S}}$ and $\mathcal{G}\brac{\mathcal{V} \setminus \mathcal{S}}$ as presented in the different claims later in Appendix \ref{sec:analysis} are intact and the rest of the proof follows exactly. Hence we can ignore this step in the analysis of our algorithm.\qed
\end{proof}
Let $A$ denote the adjacency matrix of the input graph before the action of the adversary (before Step \ref{step:seven}) and $\bar{A}$ denote the same after the action of the adversary. Due to the Lemma \ref{lem:adversary}, we can work with the adjacency matrix $A$ in the rest of the paper.

\subsection{Organization}
We present the introduction with all the relevant related work in Section \ref{sec:intro}, we shift the main analysis (due to page limit and also because it has a high overlap with the work in Khanna and Louis \cite{khanna_et_al:LIPIcs.FSTTCS.2020.27}) to Appendix \ref{sec:analysis} while we show the recovery part of the clique to Section \ref{sec:recovery} and the conclusion (Section \ref{sec:conclusion}) next.

\section{Recovering the planted clique}
\label{sec:recovery}
In the Appendix \ref{sec:analysis}, we show that under some mild conditions over the input parameters (namely when, $p$ is ``large'' and $\psi$ is ``small'') with high probability (over the randomness of the input), we have,
\begin{equation}
    \label{eq:main_eq}
    {\rm I\!E}_{i \sim \mathcal{S}} \norm{X_i}^2 \geq 1 - \psi.
\end{equation}
$\psi$ is defined in Definition $\ref{def:psi}$. We define a vertex set $\mathcal{T} \defeq \{i \in \mathcal{V} : \norm{X_i}^2 \geq 1 - \alpha \psi \}$ where $1 < \alpha < \ffrac{1}{\psi}$ is a parameter to be chosen later.

We will next show that for a cleverly chosen value of $\alpha$, we can show that $\mathcal{T}$ is also a clique, and using the facts that $\abs{\mathcal{T} \cap \mathcal{S}} > 0$ and that the boundary of the subgraph $\mathcal{G}\brac{\mathcal{S}}$ is random, we further show that $\mathcal{T} \subseteq \mathcal{S}$. Once we have established this, it is easy to recover the rest of the vertices of $\mathcal{S} \setminus \mathcal{T}$ using a simple greedy heuristic. Before that, we recall an important technical results from \cite{khanna_et_al:LIPIcs.FSTTCS.2020.27}.

\begin{lemma}[Restatement of Lemma 3.5 from \cite{khanna_et_al:LIPIcs.FSTTCS.2020.27}]
    \label{lem:size_T}
    With high probability (over the randomness of the input),
    $\abs{\mathcal{T} \cap \mathcal{S}} \geq \left(1 - \dfrac{1}{\alpha}\right)k.$
\end{lemma}
The next Lemma \ref{lem:clique1} is perhaps the most important technical result of this paper.
\begin{definition}
    \label{def:psi}
    Let $\psi$ be the function over the input parameters defined as, \\$\psi \defeq \dfrac{4\xi^2(np)(r+t+2)}{k^2\paren{1-6p-2\gamma-\dfrac{d}{s}-\dfrac{\lambda}{k}}^2}$ for the sake of brevity.
\end{definition}
\begin{lemma}
    \label{lem:clique1}
    For $\alpha = \ffrac{1}{(3\sqrt{\psi})}$ and $\psi \in (0, 1/9)$. With high probability (over the randomness of the input), the subgraph $\mathcal{G}\brac{\mathcal{T}}$ is a clique and moreover, $\mathcal{T} \subseteq \mathcal{S}$.
\end{lemma}
\begin{proof}
    By applying Lemma \ref{lem:lb_exp_xi_xj_gen} (Part 1) to the set $\mathcal{T}$, we get for all $i, j \in \mathcal{T}: \inprod{X_i, X_j} \geq 1-3\alpha\psi$. We set $\alpha$ such that $1-3\alpha\psi > 0 \iff \alpha < \ffrac{1}{(3\psi)}$. Thus we can set $\alpha = \ffrac{1}{(3\sqrt{\psi})}$. It does satisfy the bounds on $\alpha$, namely $\alpha \in (1, 1/\psi)$ when $\psi \in (0, 1/9)$. By the SDP constraints, $\inprod{X_i, X_j} = 0~\forall \paren{i, j} \notin E$ (the extra added constraint, or, Equation \ref{eq:extra_constraint}), we have that the subgraph induced on $\mathcal{T}$ is a clique. This is easy to see. Consider any two vertices $u, v \in \mathcal{T}$ such that there is no edge between $u$ and $v$, then by the above SDP constraint, $\inprod{X_u, X_v} = 0$, however by the definition of set $\mathcal{T}$, $\inprod{X_u, X_v} > 0$. This is a contradiction and thus $\mathcal{T}$ is indeed a clique.\\
    Next we prove that w.h.p. $\mathcal{T} \subseteq \mathcal{S}$. By Lemma \ref{lem:size_T}, $\abs{\mathcal{T} \cap \mathcal{S}} \geq \left(1 - \dfrac{1}{\alpha}\right)k = \paren{1-3\sqrt{\psi}}k > 0$ when $\psi \in (0, 1/9)$.
    \begin{align*}
        \therefore{\rm I\!P} \brac{\mathcal{T} \subsetneq \mathcal{S}} &\leq {\rm I\!P} \brac{\exists v \in \mathcal{V} \setminus \mathcal{S} \text{ which has an edge with all the vertices of } \mathcal{T} \cap \mathcal{S}} \\ &\leq n p^{\abs{\mathcal{T} \cap \mathcal{S}}}\leq np^{\paren{1-3\sqrt{\psi}}k}= o(1).
    \end{align*}
    where we used the union bound in step 2 and the lower bound on $\abs{\mathcal{T} \cap \mathcal{S}}$ in step 3. The step 4 holds when $p, k$ is ``large'' and $\psi$ is ``small''.
\end{proof}
We now have all the ingredients to prove our main result.
\begin{proof}[Proof of Theorem \ref{thm:main}]
    By Lemma \ref{lem:clique1} we showed that $\mathcal{T} \subseteq \mathcal{S}$, now we can use a greedy strategy to recover the rest of $\mathcal{S}$. We iterate over all vertices in $\mathcal{V} \setminus \mathcal{T}$ and add them to our set if it has edges to all of $\mathcal{T}$. A calculation similar to the one shown above can be used to ensure that no vertex of $\mathcal{V} \setminus \mathcal{S}$ enters in this greedy step.
    Also note that 
    $
        \alpha\psi = \dfrac{\psi}{3\sqrt{\psi}} = \dfrac{\sqrt{\psi}}{3}. \text{ We define } \nu \defeq 9\psi.
    $
    Here $\nu$ is nothing but a normalization of $\psi$ for a cleaner representation. We summarize this in Algorithm \ref{alg:cap} below. It is easy to see that the output of this algorithm, the set $\mathcal{Q}$ is essentially the planted clique $\mathcal{S}$ itself.
\end{proof}
\begin{algorithm}[H]
    \caption{Algorithm to recover $\mathcal{S}.$}\label{alg:cap}
        \begin{algorithmic}[1]
            \Require An Instance of \cliqueparams.
            \Ensure A vertex set $\mathcal{Q}$.
            \State Solve SDP \ref{sdp:dks} to get the vectors $\set{\set{X_i}_{i=1}^{n}, I}$.
            \State Let $\mathcal{T} = \set{i \in \mathcal{V} : \norm{X_i}^2 \geq 1-\paren{\sqrt{\nu}/9}}$. 
            \State Initialize $\mathcal{Q} = \mathcal{T}$.
            \For{vertex $v \in \mathcal{V} \setminus \mathcal{T}$,}
                \State If $v$ shares an edge with all the vertices in $\mathcal{Q}$, then update $\mathcal{Q} = \mathcal{Q} \cup \set{v}$.
                \State Else discard $v$.
            \EndFor
            \State Return $\mathcal{Q}$.
        \end{algorithmic}
\end{algorithm}

\section{Conclusion and Future Work}
\label{sec:conclusion}
In this paper, we studied the \cliqueprob~in a semi-random model and presented an SDP-based algorithm to recover the clique exactly. 

A powerful semi-random model would have any $k$ sized induced subgraph in $\mathcal{G}[\mathcal{V} \setminus \mathcal{S}]$ have an average degree of $\gamma k$, and the goal would be to give an efficient algorithm to be still able to recover the planted clique when $k=\Omega(\sqrt{n})$ while tolerating monotone deletions. To the best of our knowledge, this problem hasn't been studied in the literature, so we pose it as an interesting open question.

\section*{Acknowledgements}
YK thanks Akash Kumar, Anand Louis, and Rameesh Paul for helpful discussions. He also thanks the anonymous reviewers for their useful comments on earlier versions of the paper. He was supported by the Ministry of Education, Government of India during his stay at IISc.

%
%
%
\bibliographystyle{splncs04}
\bibliography{mybibliography}
\newpage
\appendix
\section{Analysis}
\label{sec:analysis}
We now present the complete proofs of all the technical claims in this section even though it has a overlap with the one done in the work of Khanna and Louis \cite{khanna_et_al:LIPIcs.FSTTCS.2020.27}. We bound the SDP mass corresponding to different subgraphs. The idea is to show that the SDP \ref{sdp:dks} puts a large fraction of its total mass on the vertices of $\mathcal{G}\brac{\mathcal{S}}$. 

\noindent
 We decompose the SDP objective into multiple parts (corresponding to different subgraphs) and bound each of them separately in Section \ref{sec:bound_sdp} and then combine these bounds in Section \ref{sec:putting}.
\begin{align}
    \nonumber
    \sum\limits_{i, j \in \mathcal{V}} A_{ij}\inprod{X_i, X_j} &= \underbrace{\sum\limits_{i, j \in \mathcal{S}} A_{ij}\inprod{X_i, X_j}}_{\text{Contribution from } \mathcal{G}\brac{\mathcal{S}}} + \underbrace{2\sum_{i \in \mathcal{S}, j \in \mathcal{V} \setminus \mathcal{S}} A_{ij} \inprod{X_i, X_j}}_{\text{Contribution from } \mathcal{G}\brac{\mathcal{S}, \mathcal{V} \setminus \mathcal{S}}} \nonumber \\&+ \sum\limits_{\ell=1}^{r}\underbrace{\sum\limits_{i \in \Lambda_{\ell}, j \in (\mathcal{V} \setminus \mathcal{S}) \setminus \Lambda_{\ell}} A_{ij}\inprod{X_i,X_j}}_{\text{Contribution from } \mathcal{G}\brac{\Lambda_{\ell}, (\mathcal{V} \setminus \mathcal{S}) \setminus \Lambda_{\ell}}} \nonumber\\& + \sum_{\ell=1}^{r}\underbrace{\sum_{i, j \in \Lambda_{\ell}} A_{ij} \inprod{X_i, X_j}}_{\text{Contribution from } \mathcal{G}\brac{\Lambda_{\ell}}} + \sum\limits_{\ell=1}^{t}\underbrace{\sum\limits_{i \in \Pi_{\ell}, j \in (\mathcal{V} \setminus \mathcal{S}) \setminus \Pi_{\ell}} A_{ij}\inprod{X_i,X_j}}_{\text{Contribution from } \mathcal{G}\brac{\Pi_{\ell}, (\mathcal{V} \setminus \mathcal{S}) \setminus \Pi_{\ell}}} \nonumber\\&+ \sum_{\ell=1}^{t}\underbrace{\sum_{i, j \in \Pi_{\ell}} A_{ij} \inprod{X_i, X_j}}_{\text{Contribution from } \mathcal{G}\brac{\Pi_{\ell}}}+\label{eq:sum1} \underbrace{\sum\limits_{i, j \in \Gamma} A_{ij}\inprod{X_i, X_j}}_{\text{Contribution from } \mathcal{G}\brac{\Gamma}}.
\end{align}
Note that the $1^{\text{st}}$, $4^{\text{th}}$, and the $6^{\text{th}}$ term in the Equation \ref{eq:sum1} corresponds to the subgraphs $\mathcal{G}\brac{\mathcal{S}},~\mathcal{G}\brac{\Lambda_{i_1}}~\forall {i_1} \in [r],\text{ and }~\mathcal{G}\brac{\Pi_{i_2}}~\forall {i_2} \in [t]$ respectively while the rest of the terms (i.e. the contribution from the random subgraph) can be further split as follows. This is also called the centering trick.
\begin{align}
    \label{eq:sum_11}
    2\sum_{i \in \mathcal{S}, j \in \mathcal{V} \setminus \mathcal{S}} A_{ij} \inprod{X_i, X_j} = {2p\sum_{i \in \mathcal{S}, j \in \mathcal{V} \setminus \mathcal{S}} \inprod{X_i, X_j}}+ {2\sum_{i \in \mathcal{S}, j \in \mathcal{V} \setminus \mathcal{S}} \paren{A_{ij} - p}\inprod{X_i, X_j}}
\end{align}
\begin{align}
    \label{eq:sum_33}
    \sum\limits_{\ell=1}^{r}\sum\limits_{i \in \Lambda_{\ell}, j \in (\mathcal{V} \setminus \mathcal{S}) \setminus \Lambda_{\ell}} A_{ij}\inprod{X_i,X_j} \nonumber&= {p\sum\limits_{\ell=1}^{r}\sum\limits_{i \in \Lambda_{\ell}, j \in (\mathcal{V} \setminus \mathcal{S}) \setminus \Lambda_{\ell}} \inprod{X_i,X_j}} \\&+
    {\sum\limits_{\ell=1}^{r}\sum\limits_{i \in \Lambda_{\ell}, j \in (\mathcal{V} \setminus \mathcal{S}) \setminus \Lambda_{\ell}} (A_{ij}-p)\inprod{X_i,X_j}}
\end{align}
\begin{align}
    \label{eq:sum_44}
    \sum\limits_{\ell=1}^{t}\sum\limits_{i \in \Pi_{\ell}, j \in (\mathcal{V} \setminus \mathcal{S}) \setminus \Pi_{\ell}} A_{ij}\inprod{X_i,X_j} \nonumber&= {p\sum\limits_{\ell=1}^{t}\sum\limits_{i \in \Pi_{\ell}, j \in (\mathcal{V} \setminus \mathcal{S}) \setminus \Pi_{\ell}} \inprod{X_i,X_j}}\\&+{\sum\limits_{\ell=1}^{t}\sum\limits_{i \in \Pi_{\ell}, j \in (\mathcal{V} \setminus \mathcal{S}) \setminus \Pi_{\ell}} (A_{ij}-p)\inprod{X_i,X_j}}
\end{align}
\begin{align}
    \label{eq:sum_66}
    \sum\limits_{i, j \in \Gamma} A_{ij}\inprod{X_i, X_j} = p \sum\limits_{i, j \in \Gamma} \inprod{X_i, X_j} + \sum\limits_{i, j \in \Gamma} (A_{ij} - p)\inprod{X_i, X_j}
\end{align}
\noindent
Note that there are two kinds of terms in equations \ref{eq:sum_11}, \ref{eq:sum_33}, \ref{eq:sum_44}, and \ref{eq:sum_66} which only depends on the SDP constraints, and the second, which uses the adjacency matrix of $\mathcal{G}$. Before we proceed, we introduce a new matrix for convenience.
\begin{definition}
\label{def:matrix_b}
Let $B$ be a $n \times n$ sized centered matrix (i.e., ${\rm I\!E}[B] = 0$) defined as follows.
\[ B_{ij} \defeq \begin{cases}  0 & i,j \in \mathcal{S} \text{ or } i,j \in \Lambda_{i_1}~\forall {i_1} \in[r] \text{ or } i,j \in \Pi_{i_2}~\forall {i_2} \in[t]\\
A_{ij}-p & \textrm{otherwise} \end{cases}. \]
\end{definition}
\noindent
Definition \ref{def:matrix_b} allows us to rewrite the centered terms from Equations \ref{eq:sum_11}, \ref{eq:sum_33}, \ref{eq:sum_44}, and \ref{eq:sum_66} as follows.
\begin{align}
    \nonumber
     {2\sum_{i \in \mathcal{S}, j \in \mathcal{V} \setminus \mathcal{S}} \paren{A_{ij} - p}\inprod{X_i, X_j}} &+
     {\sum\limits_{\ell=1}^{r}\sum\limits_{i \in \Lambda_{\ell}, j \in (\mathcal{V} \setminus \mathcal{S}) \setminus \Lambda_{\ell}} (A_{ij}-p)\inprod{X_i,X_j}} \\&+ \nonumber{\sum\limits_{\ell=1}^{t}\sum\limits_{i \in \Pi_{\ell}, j \in (\mathcal{V} \setminus \mathcal{S}) \setminus \Pi_{\ell}} (A_{ij}-p)\inprod{X_i,X_j}} \\&+ {\sum\limits_{i, j \in \Gamma} (A_{ij} -p) \inprod{X_i, X_j}}
        \label{eq:sum3} 
      \\&= {\sum\limits_{i, j \in \mathcal{V}} B_{ij}\inprod{X_i,X_j}}.
\end{align}
\subsection{Bounding the SDP terms}
\label{sec:bound_sdp}
In this section, we show an upper bound on the various terms of the SDP objective.
\subsubsection{Contribution from $\mathcal{G}\brac{\mathcal{S}}$, i.e., the planted clique.}
\begin{lemma}
    \label{lem:aux_one}
    For any set $\mathcal{V'} \subseteq \mathcal{V}$,
    $ \sum\limits_{i, j \in \mathcal{V'}} A_{ij}\inprod{X_i, X_j} \leq \sum\limits_{i \in \mathcal{V'}} {\sum\limits_{j\in \mathcal{V'}}A_{ij} } \norm{X_i}^2. $
\end{lemma}
\begin{proof}
    \begin{align*}
    \sum\limits_{i, j \in \mathcal{V'}} A_{ij}\inprod{X_i, X_j} 
    & \leq  \sum\limits_{i, j \in \mathcal{V'}} A_{ij}
    \left(\dfrac{\norm{X_i}^2 + \norm{X_j}^2}{2}\right) = \sum\limits_{i \in \mathcal{V'}} {\sum_{j\in \mathcal{V'}}A_{ij} } \norm{X_i}^2.
    \end{align*}
    Here the first inequality holds by expanding, $\norm{X_i - X_j}^2 \geq 0~\forall i,j \in \mathcal{V'}$ and the second equality holds because $A_{ij} = A_{ji}~\forall i, j \in \mathcal{V}$.\qed
\end{proof}
\begin{lemma}
    \label{lem:one}
    $ \sum\limits_{i, j \in \mathcal{S}} A_{ij}\inprod{X_i, X_j} \leq k^2\paren{{\rm I\!E}_{i \sim \mathcal{S}} \norm{X_i}^2}. $
\end{lemma}
\begin{proof}
    From Lemma \ref{lem:aux_one} with $\mathcal{V'} = \mathcal{S}$.
    \begin{align*}
        \sum\limits_{i, j \in \mathcal{S}} A_{ij}\inprod{X_i, X_j} \leq \sum\limits_{i \in \mathcal{S}} {\sum_{j\in \mathcal{S}} A_{ij} } \norm{X_i}^2 
        = k^2 {\rm I\!E}_{i \sim \mathcal{S}} \norm{X_i}^2. 
    \end{align*}
    Here the last equality follows from the fact that the graph induced on $\mathcal{S}$ is a clique, so $\sum_{j\in \mathcal{S}}A_{ij}  = k$ for each $i \in \mathcal{S}$.\qed
\end{proof}
\subsubsection{Contribution from the random subgraphs.}
\begin{lemma}
    \label{lem:two}
    $ \sum\limits_{i \in \mathcal{S}, j \in \mathcal{V} \setminus \mathcal{S}} \inprod{X_i,X_j} \leq 3k^2\paren{1-{\rm I\!E}_{i \sim \mathcal{S}} \norm{X_i}^2} .$
\end{lemma}
Before proving our main results, we setup some groundwork. 
\begin{lemma}
    \label{lem:lb_exp_xi_xj_gen}
    Let $\set{\set{Y_i}_{i=1}^{n}, I_Y}$ be any feasible solution of SDP \ref{sdp:dks} and $\mathcal{V'} \subseteq \mathcal{V}$ such that, 
    \begin{enumerate}
        \item If $\norm{Y_i}^2 \geq 1-\epsilon\text{ for all } i \in \mathcal{V'}$ where $0 \leq \epsilon\leq 1$, then $\inprod{Y_i, Y_j} \geq 1-3\epsilon\text{ for all } i, j \in \mathcal{V'}$.
        \item If ${\rm I\!E}_{i \sim \mathcal{V'}}\norm{Y_i}^2 \geq 1-\epsilon$ where $0 \leq \epsilon\leq 1$, then ${\rm I\!E}_{i,j \sim \mathcal{V'}} \inprod{Y_i, Y_j} \geq 1-4\epsilon$.
    \end{enumerate}
\end{lemma}
\begin{proof}
    We first introduce vectors $Z_i \in {\rm I\!R}^{n+1}$ and scalars $\alpha_i \in {\rm I\!R}$ (for all $i \in \mathcal{V'}$) such that $Y_i = \alpha_i I_Y + Z_i$ and $\inprod{I_Y, Z_i} = 0$. 
    Using SDP constraint \ref{eq:sdp7} we get 
    \begin{align*}
    \norm{Y_i}^2 = \inprod{Y_i, I_Y} = \inprod{\alpha_i I_Y+ Z_i, I_Y} = \alpha_i \inprod{I_Y,I_Y} + \inprod{I_Y,Z_i} = \alpha_{i}.
    \end{align*}
    Next,
    \begin{align}
    \label{eq:y_val}
    \norm{Y_i}^2 = \alpha_i^2\norm{I}^2 + \norm{Z_i}^2 = \norm{Y_i}^4 + \norm{Z_i}^2
    \implies \norm{Z_i} = \sqrt{\norm{Y_i}^2 - \norm{Y_i}^4}.
    \end{align}
    For $i, j \in \mathcal{V'}$,
    \begin{align*}
    \inprod{Y_i, Y_j} &= \inprod{\norm{Y_i}^2I_Y+Z_i,\norm{Y_j}^2I_Y+Z_j}\\
    &= \norm{Y_i}^2\norm{Y_j}^2 \inprod{I_Y, I_Y} + \norm{Y_i}^2\inprod{I_Y, Z_j} + \norm{Y_j}^2\inprod{I_Y, Z_i} + \inprod{Z_i, Z_j}\\ 
    &= \norm{Y_i}^2\norm{Y_j}^2 + \inprod{Z_i, Z_j} \qquad (\because \inprod{I_Y, Z_i} = 0)\\
    &\geq \norm{Y_i}^2\norm{Y_j}^2 - \norm{Z_i}\norm{Z_j} \qquad (\text{since the max. angle can be }\pi)\\
    &= \norm{Y_i}^2\norm{Y_j}^2 - \left(\sqrt{\norm{Y_i}^2 - \norm{Y_i}^4}\right)\left(\sqrt{\norm{Y_j}^2 - \norm{Y_j}^4}\right)~~~~ (\text{by eqn } \ref{eq:y_val})
    \\&= \norm{Y_i}^2 \norm{Y_j}^2 - \paren{\sqrt{\norm{Y_i}^2(1 - \norm{Y_i}^2)}}\paren{\sqrt{\norm{Y_j}^2(1 - \norm{Y_j}^2)}}.
    \end{align*}
    \begin{enumerate}
        \item
        Since $\norm{Y_i}^2 \geq 1 - \epsilon$ using this in above equation we get
        \begin{align*}
       \inprod{Y_i,Y_j} &  \geq (1 - \epsilon)^2 - \sqrt{(\epsilon)}\sqrt{(\epsilon)} \qquad(\because \norm{Y_i}^2 \leq 1 \text{ and } 1 - \norm{Y_i}^2 \leq \epsilon ) \\ 
        & = (1 - \epsilon)^2 - \epsilon = 1 + \epsilon^2 -2\epsilon - \epsilon \geq 1 -3\epsilon .
        \end{align*}
        \item
        Summing both sides $\forall i, j \in \mathcal{V'}$ and dividing by $\abs{\mathcal{V'}}^2$,\\
        $
        \sum\limits_{i, j \in \mathcal{V'}} \dfrac{\inprod{Y_i, Y_j}}{\abs{\mathcal{V'}}^2} 
        \geq 
        \left(\sum\limits_{i \in \mathcal{V'}} \dfrac{\norm{Y_i}^2}{\abs{\mathcal{V'}}}\right)\left(\sum\limits_{j \in \mathcal{V'}} \dfrac{\norm{Y_j}^2}{\abs{V'}}\right)\\-\left(\sum\limits_{i \in \mathcal{V'}}\dfrac{\sqrt{\norm{Y_i}^2 \norm{Y_i}^4}}{\abs{\mathcal{V'}}}\right)
        \left(\sum\limits_{j \in \mathcal{V'}} \dfrac{\sqrt{\norm{Y_j}^2 - \norm{Y_j}^4}}{\abs{\mathcal{V'}}}\right).
        $
        \begin{align*}
        \therefore{\rm I\!E}_{i, j \sim \mathcal{V'}} \inprod{Y_i, Y_j} &\geq \left({\rm I\!E}_{i \sim \mathcal{V'}} \norm{Y_i}^2\right)^2 - \left({\rm I\!E}_{i \sim \mathcal{V'}} \sqrt{\norm{Y_i}^2 - \norm{Y_i}^4}\right)^2\\
        &\geq \left({\rm I\!E}_{i \sim \mathcal{V'}} \norm{Y_i}^2\right)^2 - \left({\rm I\!E}_{i \sim \mathcal{V'}} [\norm{Y_i}^2 - \norm{Y_i}^4]\right)\\ 
        &\qquad\qquad\qquad\paren{\text{by Jensen's inequality}}\\
        &\geq \left({\rm I\!E}_{i \sim \mathcal{V'}} \norm{Y_i}^2\right)^2 - {\rm I\!E}_{i \sim \mathcal{V'}} \norm{Y_i}^2 + \left({\rm I\!E}_{i \sim \mathcal{V'}} \norm{Y_i}^2\right)^2 \\&\qquad\qquad\qquad\paren{\because {\rm I\!E}\norm{Y_i}^4 \geq \paren{{\rm I\!E}\norm{Y_i}^2}^2}\\
        &= 2\left({\rm I\!E}_{i \sim \mathcal{V'}} \norm{Y_i}^2\right)^2 - {\rm I\!E}_{i \sim \mathcal{V'}} \norm{Y_i}^2
        \geq 2\left(1-\epsilon\right)^2 - 1 \\&\qquad\qquad\qquad\paren{\because {\rm I\!E}\norm{Y_i}^2  \leq 1} \\
        & = 1 - 4 \epsilon+ 2 \epsilon^2 \geq 1 - 4\epsilon.
        \end{align*}\qed
    \end{enumerate}
\end{proof}
\begin{lemma}
    \label{lem:lb_exp_xi_xj_s}
    $ {\rm I\!E}_{i,j \sim \mathcal{S}}\inprod{X_i,X_j} \geq 4{\rm I\!E}_{i \sim \mathcal{S}} \norm{X_i}^2 - 3 .$
\end{lemma}
\begin{proof}
    Using Lemma \ref{lem:lb_exp_xi_xj_gen} (2) on the set $\mathcal{S}$ and with $\epsilon= 1-{\rm I\!E}_{i \sim \mathcal{S}}\norm{X_i}^2$, we get the lower bound $1-4\paren{1-{\rm I\!E}_{i \sim \mathcal{S}}\norm{X_i}^2} = 4{\rm I\!E}_{i \sim \mathcal{S}} \norm{X_i}^2 - 3$.\qed
\end{proof}
We are now ready to prove Lemma \ref{lem:two}.
\begin{proof}[Proof of Lemma \ref{lem:two}]
    \begin{align*}
    \sum\limits_{i \in \mathcal{S}, j \in \mathcal{V} \setminus \mathcal{S}} \inprod{X_i,X_j} 
    & = \sum\limits_{i \in \mathcal{S}, j \in \mathcal{V}} \inprod{X_i,X_j}  - \sum\limits_{i \in \mathcal{S}, j \in \mathcal{S}} \inprod{X_i,X_j} \\ 
    & \leq k\sum\limits_{i \in \mathcal{S}} \norm{X_i}^2 - \sum\limits_{i \in \mathcal{S}, j \in \mathcal{S}} \inprod{X_i,X_j} \qquad (\text{by SDP constraint \ref{eq:sdp5}})\\ 
    &= k^2\paren{{\rm I\!E}_{i \sim \mathcal{S}}\norm{X_i}^2} - k^2\paren{{\rm I\!E}_{i, j \sim \mathcal{S}} \inprod{X_i,X_j}}\\
    &\leq k^2\paren{{\rm I\!E}_{i \sim \mathcal{S}}\norm{X_i}^2} - k^2\paren{4{\rm I\!E}_{i \sim \mathcal{S}} \norm{X_i}^2 - 3} \qquad (\text{by Lemma \ref{lem:lb_exp_xi_xj_s}})\\
    &= 3k^2\paren{1-{\rm I\!E}_{i \sim \mathcal{S}}\norm{X_i}^2} .
    \end{align*}\qed
\end{proof}
\begin{lemma}
    \label{lem:three_3}
    \begin{enumerate}
        \item 	$ \sum\limits_{\ell=1}^{r}\sum\limits_{i \in \Lambda_{\ell}, j \in (\mathcal{V} \setminus \mathcal{S}) \setminus \Lambda_{\ell}} \inprod{X_i,X_j} \leq k^2\paren{1-{\rm I\!E}_{i \sim \mathcal{S}} \norm{X_i}^2} . $
        \item 	$ \sum\limits_{\ell=1}^{t}\sum\limits_{i \in \Pi_{\ell}, j \in (\mathcal{V} \setminus \mathcal{S}) \setminus \Pi_{\ell}} \inprod{X_i,X_j} \leq k^2\paren{1-{\rm I\!E}_{i \sim \mathcal{S}} \norm{X_i}^2} . $
        \item 	$ \sum\limits_{i, j \in \Gamma} \inprod{X_i,X_j} \leq k^2\paren{1-{\rm I\!E}_{i \sim \mathcal{S}} \norm{X_i}^2} . $
    \end{enumerate}
\end{lemma}
\begin{proof}[Proof of Part 1]
    Note that for all $\ell \in [r]$,
    \begin{align*}
    \sum\limits_{i \in \Lambda_\ell, j \in (V \setminus S)\setminus \Lambda_{\ell}} \inprod{X_i,X_j} 
    &\leq \sum\limits_{i \in \Lambda_\ell, j \in \mathcal{V}} \inprod{X_i,X_j} \leq k\sum\limits_{i \in \Lambda_\ell} \inprod{X_i,X_i}.
    \end{align*}
    The first inequality just follows from the SDP constraint \ref{eq:sdp5} (non-negativity) and the second one follows from the SDP constraint \ref{eq:sdp4}. Summing up for all $\ell \in [r]$,
    \begin{align*}
        \sum\limits_{\ell=1}^{r}\sum\limits_{i \in \Lambda_{\ell}, j \in (\mathcal{V} \setminus \mathcal{S}) \setminus \Lambda_{\ell}} \inprod{X_i,X_j} \leq k\sum_{\ell = 1}^{r}\sum\limits_{i \in \Lambda_\ell} \inprod{X_i,X_i} &\leq k\sum\limits_{i \in \mathcal{V} \setminus \mathcal{S}} \inprod{X_i,X_i} \\&\leq k^2\paren{1-{\rm I\!E}_{i \sim \mathcal{S}} \norm{X_i}^2} \\&\paren{\text{By SDP constraint \ref{eq:sdp2}}}.
    \end{align*}
    The proof of part 2 and 3 follows similarly.\qed
\end{proof}
\begin{lemma}
    \label{lem:four}
    $ \sum\limits_{i \in \mathcal{S}, j \in \mathcal{V} \setminus \mathcal{S}} B_{ij}\inprod{X_i,X_j} \leq 
    \norm{B}\sqrt{\sum\limits_{i \in \mathcal{S}} \norm{X_i}^2}\sqrt{\sum\limits_{i \in \mathcal{V} \setminus \mathcal{S}} \norm{X_i}^2} . $
\end{lemma}
\begin{proof}
    Recall that w.l.o.g., we can assume that the SDP vectors to be of dimension $n+1$. We define two matrices $Y, Z$ each of size $(n+1) \times n$. For all $i \in \mathcal{S}$, the vector $X_i$ is placed at the $i^{th}$ column of the matrix $Y$ while the rest of the entries of $Y$ are zero. Similarly for all $j \in \mathcal{V} \setminus \mathcal{S}$, the vector $X_j$ is placed at the $j^{th}$ column of the matrix $Z$ and rest of the entries of $Z$ are zero. We use $Y_i$ to denote the $i^{th}$ column vector of the matrix $Y$. Similarly, $Y^{T}_{j}$ denotes the $j^{th}$ column vector of the matrix $Y^T$.
    \begin{align*}
    \sum\limits_{i \in \mathcal{S}, j \in \mathcal{V} \setminus {\mathcal{S}}} B_{ij}\inprod{X_{i},X_{j}} 
    &\leq \sum\limits_{i, j \in \mathcal{V}}\sum\limits_{l=1}^{n+1} B_{ij}X_i(l)X_j(l) = \sum\limits_{l=1}^{n+1}\sum\limits_{i, j \in \mathcal{V}} B_{ij}X_i(l)X_j(l) \\
    &= 2\sum\limits_{l=1}^{n+1} \paren{Y^{T}_{l}}^T B \paren{Z^{T}_{l}}
    \leq 2\sum\limits_{l=1}^{n+1} \norm{Y^{T}_{l}}\norm{Z^{T}_{l}}\norm{B} \\
    \nonumber
    &\leq 2\norm{B}\sqrt{\sum\limits_{l=1}^{n+1} \norm{Y^{T}_{l}}^2}\sqrt{\sum\limits_{l=1}^{n+1} \norm{Z^{T}_{l}}^2} \\&\qquad\qquad (\text{by Cauchy-Schwarz inequality})\\
    &= 2\norm{B}\sqrt{\sum\limits_{i \in \mathcal{S}} \norm{X_i}^2}\sqrt{\sum\limits_{i \in \mathcal{V} \setminus \mathcal{S}} \norm{X_i}^2} \\&\qquad\qquad (\text{rewriting entries using columns}).
    \end{align*}\qed
\end{proof}
We start by stating a standard result on the spectral norm of random matrices.
\begin{theorem}[\cite{7282694}, Lemma 30; Wigner's Bound]
    \label{thm:ub_spectral_norm_gen}
    Let $M$ be a symmetric matrix of size $n \times n$ with zero diagonals and independent entries such that $M_{ij} = M_{ji} \sim Bern\paren{p_{ij}}$ for all $i<j$ with $p_{ij} \in [0,1]$. Assume $p_{ij}\paren{1-p_{ij}} \leq r$ for all $i < j$ and $nr = \Omega\paren{\log n}$. Then, with high probability (over the randomness of matrix $M$),
    $ \norm{M - {\rm I\!E}\brac{M}} \leq {\bigO\paren{1}}\sqrt{nr} .$
\end{theorem}
\begin{lemma}
    \label{lem:four_1}
    There exists universal constants $\kappa, \xi \in \mathbb{R}^{+}$ such that if \\$p \in \left[\dfrac{\kappa \log n}{n}, 1\right)$, then
    $\norm{B} \leq \xi\sqrt{np}$ with high probability (over the randomness of the input).
\end{lemma}

\begin{proof}
    Let $H$ be the adjacency matrix (symmetric) of the random subgraph of $\mathcal{G}$, i.e., without any edges inside $\mathcal{G}\brac{\mathcal{S}},~\mathcal{G}\brac{\Lambda_{i_1}}~\forall {i_1} \in [r],\text{ and }~\mathcal{G}\brac{\Pi_{i_2}}~\forall {i_2} \in [t]$. Therefore, for $i, j \in \mathcal{S}$, $i, j \in {\Lambda_{i_1}}~\forall i_1 \in [r]$, or $i, j \in {\Pi_{i_2}}~\forall i_2 \in [t]$, $H_{ij}$ is identically $0$. We know that all the entries of $H$ are independent because of the assumption of random edges being added independently. By definition, $H_{ij}$ is sampled from the Bernoulli distribution with parameter $p$ or $H_{ij} \sim Bern\paren{p}$. For the parameter range $p \in \left[\ffrac{\probabilityconstant \log n}{n}, 1\right)$, we have $p(1-p) \leq p$ and $np = \Omega(\log n)$. We now apply Theorem \ref{thm:ub_spectral_norm_gen} to matrix $H$ with the parameter $r = p$ to get,
    $
    \norm{H - {\rm I\!E}\brac{H}} \leq {\bigO\paren{1}}\sqrt{np} = \matrixconstant\sqrt{np} .
    $
    where $\matrixconstant \in \mathbb{R}^{+}$ is the constant from Theorem \ref{thm:ub_spectral_norm_gen}. By definition, we have $B = H - {\rm I\!E}\brac{H}$. Thus,
    $\norm{B} = \norm{H - {\rm I\!E}\brac{H}} \leq \matrixconstant\sqrt{np}.$ \qed
\end{proof}

\begin{lemma}
    \label{lem:five}
    With high probability (over the randomness of the input),
    \[ \sum\limits_{i, j \in \mathcal{V}} B_{ij}\inprod{X_i,X_j} \leq 
     2\xi k \sqrt{np} {\sqrt{(r+t+1)}}\sqrt{\paren{1-{\rm I\!E}_{i \sim S} \norm{X_i}^2}}\]
     if $p \in \left[\dfrac{\kappa \log n}{n}, 1\right)$, where $\kappa, \xi \in {\rm I\!R}^+$ are a universal constants.
\end{lemma}
\begin{proof}
    A similar calculation to the one done in Lemma \ref{lem:four}, we can easily show that, $\forall i_1 \in [r]$,
    \begin{align*}
    \sum\limits_{i \in \Lambda_{i_1}, j \in (\mathcal{V} \setminus \mathcal{S}) \setminus \Lambda_{i_1}} B_{ij}\inprod{X_{i},X_{j}} 
    &\leq \norm{B}\sqrt{\sum\limits_{i \in \Lambda_{i_1}} \norm{X_i}^2}\sqrt{\sum\limits_{i \in (\mathcal{V} \setminus \mathcal{S}) \setminus \Lambda_{i_1}} \norm{X_i}^2} \\&\leq \norm{B}\sqrt{\sum\limits_{i \in \Lambda_{i_1}} \norm{X_i}^2}\sqrt{\sum\limits_{i \in \mathcal{V} \setminus \mathcal{S}} \norm{X_i}^2}.
    \end{align*}
    $\forall i_2 \in [t]$,
    \begin{align*}
    \sum\limits_{i \in \Pi_{i_2}, j \in (\mathcal{V} \setminus \mathcal{S}) \setminus \Pi_{i_2}} B_{ij}\inprod{X_{i},X_{j}} 
    &\leq \norm{B}\sqrt{\sum\limits_{i \in \Pi_{i_2}} \norm{X_i}^2}\sqrt{\sum\limits_{i \in (\mathcal{V} \setminus \mathcal{S}) \setminus \Pi_{i_2}} \norm{X_i}^2} \\&\leq \norm{B}\sqrt{\sum\limits_{i \in \Pi_{i_2}} \norm{X_i}^2}\sqrt{\sum\limits_{i \in \mathcal{V} \setminus \mathcal{S}} \norm{X_i}^2}.
    \end{align*}
    And,
    \begin{align*}
    \sum\limits_{i, j \in \Gamma}\inprod{X_{i},X_{j}} \leq \norm{B}\sqrt{\sum\limits_{i \in \Gamma} \norm{X_i}^2}\sqrt{\sum\limits_{i \in \mathcal{V} \setminus \mathcal{S}} \norm{X_i}^2}.
    \end{align*}
    Summing up for $\mathcal{S}, \Lambda_{i_1}'s, \Pi_{i_2}'s \text{ and } \Gamma$,
    \begin{align*}
        \sum\limits_{i, j \in \mathcal{V}} B_{ij}\inprod{X_i,X_j} &\leq 2\norm{B}\paren{\sqrt{\sum\limits_{i \in \mathcal{S}} \norm{X_i}^2} +  \sum_{i_1=1}^{r}\sqrt{\sum\limits_{i \in \Lambda_{i_1}} \norm{X_i}^2} + \sum_{i_2=1}^{t}\sqrt{\sum\limits_{i \in \Pi_{i_2}} \norm{X_i}^2} + \sqrt{\sum\limits_{i \in \Gamma} \norm{X_i}^2}}\\&\qquad\times\sqrt{\sum\limits_{i \in \mathcal{V} \setminus \mathcal{S}} \norm{X_i}^2} \\
        &\leq 2\norm{B}{\sqrt{(r+t+2)\paren{\sum\limits_{i \in \mathcal{S}} \norm{X_i}^2 + \sum_{i_1 = 1}^{r}\sum\limits_{i \in \Lambda_{i_1}} \norm{X_i}^2 + \sum_{i_2=1}^{t}\sum\limits_{i \in \Pi_{i_2}} \norm{X_i}^2 +\sum\limits_{i \in \Gamma} \norm{X_i}^2}}}\\&\qquad\times\sqrt{\sum\limits_{i \in \mathcal{V} \setminus \mathcal{S}} \norm{X_i}^2}\qquad\qquad(\text{by Cauchy-Schwarz Inequality})\\
        &= 2\norm{B}{\sqrt{(r+t+2)k}}\sqrt{k\paren{1-{\rm I\!E}_{i \sim \mathcal{S}} \norm{X_i}^2}}~~(\text{by SDP constraint \ref{eq:sdp2}})\\
        &\leq 2\xi k \sqrt{np} {\sqrt{(r+t+2)}}\sqrt{\paren{1-{\rm I\!E}_{i \sim \mathcal{S}} \norm{X_i}^2}}~~(\text{by Lemma \ref{lem:four_1}}).
    \end{align*} \qed
\end{proof}

\subsubsection{Contribution from $\mathcal{G}\brac{\Lambda_\ell}'s$ i.e. expander graphs.}
We bound the SDP terms by using an important result from Bhaskara \etal\cite{10.1145/1806689.1806719}.

\begin{lemma}[\cite{10.1145/1806689.1806719}, Theorem 6.1]
    \label{lem:eight}
    For a $(d', \lambda)$-expander graph on $n$ vertices, the value of the SDP \ref{sdp:dks} is at most $\dfrac{k^2d'}{n} + k\lambda.$
\end{lemma}

\begin{lemma}
    \label{lem:nine}
    $
        \sum\limits_{\ell=1}^{r}\sum\limits_{i, j \in \Lambda_{\ell}} A_{ij} \inprod{X_i, X_j} \leq k^2\paren{\dfrac{d}{s}+\dfrac{\lambda}{k}}\paren{1-{\rm I\!E}_{i \sim \mathcal{S}} \norm{X_i}^2}.
    $
\end{lemma}
\begin{proof}
    Summing up for all $\ell \in [r]$ and using Lemma \ref{lem:eight} for each sum and a scaling factor of $\sum\limits_{i \in \Lambda_\ell} \norm{X_i}^2/k$, we get,
    \begin{align*}
        \sum_{\ell=1}^{r}\sum_{i, j \in \Lambda_{\ell}} A_{ij} \inprod{X_i, X_j} \leq \sum_{\ell=1}^{r}\paren{\dfrac{kd}{s}+\lambda}\paren{\sum_{i \in \Lambda_\ell} \norm{X_i}^2} \leq
        \paren{\dfrac{kd}{s}+\lambda}\paren{\sum_{i \in \mathcal{V} \setminus \mathcal{S}} \norm{X_i}^2} &\\= k^2\paren{\dfrac{d}{s}+\dfrac{\lambda}{k}}\paren{1-{\rm I\!E}_{i \sim \mathcal{S}} \norm{X_i}^2}
    \end{align*}
    where we used SDP constraint \ref{eq:sdp2} in the third step.\qed
\end{proof}

\paragraph{Contribution from $\mathcal{G}\brac{\Pi_\ell}'s$ i.e. graphs with low average degree.}

We restate the assumption on $\mathcal{G}\brac{\Pi_{\ell}}~\forall {\ell} \in [t]$ in our model for clarity,
$\max\limits_{\mathcal{W} \subseteq \Pi_\ell}\left\{\dfrac{\sum\limits_{i, j \in \mathcal{W}}A_{ij}}{2\abs{\mathcal{W}}}\right\} \leq \gamma k.$
Given a graph $\mathcal{H} = (\mathcal{V'}, \mathcal{E'})$, consider the following LP relaxation for the problem of computing $\max\limits_{\mathcal{V''} \subseteq \mathcal{V'}} \sum\limits_{i, j \in \mathcal{V''}}C_{ij}/\Abs{\mathcal{V''}}$ where $C$ is the adjacency matrix of graph $\mathcal{H}$.
\begin{LP}
    \label{lp:lp}
    \begin{align}
    \max_{\mathcal{V''} \subseteq \mathcal{V'}}\qquad\qquad\qquad\qquad\quad
    \label{eq:lp0}
    \sum\limits_{\set{i,j} \in \mathcal{E'}} &C_{ij}x_{ij} \\
    \text{subject to}\qquad\qquad\qquad\qquad\qquad\qquad
    \label{eq:lp1}
    &x_{i j} \leq y_{i} & \forall \set{i, j} \in \mathcal{E'} \\
    \label{eq:lp3}
    &\sum_{i \in \mathcal{V'}} y_{i} \leq 1 \\
    \label{eq:lp4}
    &x_{i j} \geq 0 & \forall \set{i, j} \in \mathcal{E'} \\
    \label{eq:lp5}
    &y_{i}  \geq 0 & \forall i \in \mathcal{V'}
    \end{align}
\end{LP}
\noindent
Charikar \cite{10.1007/3-540-44436-X_10} proved the following result.
\begin{theorem}
[\cite{10.1007/3-540-44436-X_10}, Theorem 1]
    \label{lem:charikar_lp}
    For a given graph $H = (\mathcal{V'}, \mathcal{E'})$ with adjacency matrix $C$,
    $\max\limits_{\mathcal{V''} \subseteq \mathcal{V'}} \dfrac{\sum\limits_{i, j \in \mathcal{V''}}C_{ij}}{|\mathcal{V''}|} = \OPT\paren{LP}$
    where $\OPT\paren{LP}$ denotes the optimal value of the LP \ref{lp:lp}.
\end{theorem}
\noindent
Note that both the constraints of LP \ref{lp:lp} and SDP \ref{sdp:dks} closely resemble each other and hence after an appropriate scaling, we can construct a feasible solution to LP \ref{lp:lp} using our SDP solution.
\begin{lemma}
    \label{lem:lpfeasibility}
    For $\mathcal{G}\brac{\Pi_\ell}$ $\forall \ell \in [t]$,
    $
    x_{ij} \defeq \dfrac{\inprod{X_i, X_j}}{\sum\limits_{i \in \Pi_\ell} \norm{X_i}^2}
    \text{ for } \set{i, j} \in \mathcal{E}\paren{\mathcal{G}\brac{\Pi_\ell}} ~\text{and}~
    y_{i} \defeq \dfrac{\norm{X_i}^2}{\sum\limits_{i \in \Pi_\ell} \norm{X_i}^2}
    \text{ for } i \in \Pi_\ell
    $
    is a feasible solution for LP \ref{lp:lp}.
\end{lemma}
\begin{proof}
    \begin{enumerate}
        \item  The LP constraints \ref{eq:lp1} and the non-negativity constraints \ref{eq:lp4} and \ref{eq:lp5} hold by the SDP constraint \ref{eq:sdp5}.
        \item For LP constraint \ref{eq:lp3}, 
        \\$\sum\limits_{i \in \Pi_\ell} y_i = \sum\limits_{i \in \Pi_\ell} \left(\dfrac{\norm{X_i}^2}{\sum\limits_{i \in \Pi_\ell} \norm{X_i}^2} \right)= \dfrac{\sum\limits_{i \in \Pi_\ell} \norm{X_i}^2}{\sum\limits_{i \in \Pi_\ell} \norm{X_i}^2} = 1$
    \end{enumerate}\qed
\end{proof}
\begin{lemma}
    \label{lem:lpfeasibility2}
	For $\mathcal{G}\brac{\Pi_\ell}$ $\forall \ell \in [t]$,
	$ \sum\limits_{i, j \in \Pi_\ell} A_{ij}\inprod{X_i,X_j} \leq 2\left(\gamma k\right) \sum\limits_{i \in \Pi_\ell} \norm{X_i}^2 .$
\end{lemma}
\begin{proof}
	By Lemma \ref{lem:lpfeasibility}, we know that any feasible solution to LP \ref{lp:lp} 
	satisfies
	\[ \sum\limits_{\set{i,j} \in \mathcal{E'}} C_{ij}x_{ij} \leq \max_{\mathcal{V}'' \subseteq \mathcal{V'}} \frac{\sum\limits_{i, j \in \mathcal{V}''}C_{ij}}{\Abs{\mathcal{V}''}} . \]
	\begin{align*}
	\therefore\sum\limits_{i, j \in \Pi_\ell} A_{ij}\inprod{X_i, X_j} &= 2\left(\sum\limits_{\set{i, j} \in \mathcal{E}\paren{\mathcal{G}\brac{\Pi_\ell}}} \dfrac{A_{ij}\inprod{X_i, X_j}}{\sum\limits_{i \in \Pi_\ell} \norm{X_i}^2}\right) \sum\limits_{i \in \Pi_\ell} \norm{X_i}^2 \\
	& \leq 2\left(\max_{\mathcal{W} \subseteq \Pi_\ell} \frac{\sum\limits_{i, j \in \mathcal{W}}C_{ij}}{\abs{\mathcal{W}}}\right) \sum\limits_{i \in \Pi_\ell} \norm{X_i}^2 
	\leq 2\left(\gamma k\right) \sum\limits_{i \in \Pi_\ell} \norm{X_i}^2 .
	\end{align*}\qed
\end{proof}
\begin{lemma}
	\label{lem:seven}
	$
	\sum\limits_{\ell=1}^{t}\sum\limits_{i, j \in \Pi_{\ell}} A_{ij} \inprod{X_i, X_j} \leq 2\gamma k^2 \paren{1-{\rm I\!E}_{i \sim \mathcal{S}} \norm{X_i}^2}.
	$
\end{lemma}
\begin{proof}
	Summing up for all $\ell \in [t]$ and using Lemma \ref{lem:lpfeasibility2} for each sum, we get,
	\begin{align*}
	       \sum_{\ell=1}^{t}\sum_{i, j \in \Pi_{\ell}} A_{ij} \inprod{X_i, X_j} &\leq 2\gamma k \sum_{i \in \mathcal{V} \setminus \mathcal{S}} \norm{X_i}^2 \leq 2\gamma k^2 \paren{1-{\rm I\!E}_{i \sim \mathcal{S}} \norm{X_i}^2}\\ &\qquad\qquad\qquad\qquad\qquad(\text{by SDP constraint \ref{eq:sdp2}}).
	\end{align*}\qed
\end{proof}
\subsection{Putting things together}
\label{sec:putting}
In this Section, we combine the above bounds.
\begin{proposition}
	With high probability (over the randomness of the input),
	  \[
		{\rm I\!E}_{i \sim \mathcal{S}} \norm{X_i}^2 \geq 1 - \dfrac{4\xi^2(np)(r+t+2)}{k^2\paren{1-6p-2\gamma-\dfrac{d}{s}-\dfrac{\lambda}{k}}^2}
	\]
	if $p \in \left[\dfrac{\kappa \log n}{n}, 1\right)$, where $\kappa, \xi \in {\rm I\!R}^+$ are a universal constants.
\end{proposition}
\begin{proof}
Since our SDP is a maximization relaxation (See SDP \ref{sdp:dks}), we have that, 
\begin{align*}
    k^2 &\leq \sum\limits_{i, j \in \mathcal{V}} A_{ij}\inprod{X_i, X_j}\\
    k^2 &\leq k^2\paren{{\rm I\!E}_{i \sim S} \norm{X_i}^2} + 6pk^2\paren{1-{\rm I\!E}_{i \sim S}\norm{X_i}^2} \\
    &\qquad\qquad+ 2\xi k \sqrt{np} {\sqrt{(r+t+2)}}\sqrt{\paren{1-{\rm I\!E}_{i \sim S} \norm{X_i}^2}}\\ &\qquad\qquad+ k^2\paren{\dfrac{d}{s}+\dfrac{\lambda}{k}}\paren{1-{\rm I\!E}_{i \sim \mathcal{S}} \norm{X_i}^2} \\
    &\qquad\qquad+ 2\gamma k^2 \paren{1-{\rm I\!E}_{i \sim S} \norm{X_i}^2}.
\end{align*}
where we used decomposition of the sum $\sum\limits_{i, j \in \mathcal{V}} A_{ij}\inprod{X_i, X_j}$, and the results from the Section \ref{sec:bound_sdp}. Rearranging, cancelling terms, and using the fact that the function $x \rightarrow x^2,~\forall x \in {\rm I\!R}^+$  is increasing, we get,
$
    {\rm I\!E}_{i \sim \mathcal{S}} \norm{X_i}^2 \geq 1 - \dfrac{4\xi^2(np)(r+t+2)}{k^2\paren{1-6p-2\gamma-\dfrac{d}{s}-\dfrac{\lambda}{k}}^2}.	
$\qed
\end{proof}
\end{document}